\documentclass[12pt]{article}
\usepackage[margin=1in]{geometry}

\usepackage{style/qstyle}
\usepackage[normalem]{ulem}

\title{Almost perfect strategies for projection\\ games are approximately tracial}

\author{Eric Culf}
\affil{University of Waterloo}
\date{\vspace{-1cm}}

\begin{document}

\maketitle

\begin{abstract}
    Projection games constitute an important class of nonlocal games where, for any answer from the first player, there is a unique correct answer for the second player. This class of games captures nonlocal games arising from constraint satisfaction problems, oracularisations, and unique games. However, due to the asymmetry between the players, projection games are in general not synchronous, and therefore the powerful results constraining the structure of almost perfect strategies for synchronous games do not apply. In this work, we adapt results of Marrakchi and de la Salle~\cite{MdlS23} for synchronous games to show that, in both the quantum and commuting-operator models, any strategy that wins with probability $1-\varepsilon$ in a projection game gives rise to a tracial strategy that wins with probability $1-O((L\varepsilon)^{1/4})$, where $L$ is the inverse of the minimal conditional probability of a question for the second player being sampled given a question to the first. For constraint system games, this strengthens the rounding result of Paddock~\cite{Pad22} by eliminating the dependence on number of constraints and improving the dependence on constraint size, while also generalising to the commuting-operator setting.
\end{abstract}

\section{Introduction}

In a nonlocal game, a referee sends questions to non-communicating players Alice and Bob, who provide answers to the referee; the players win if the questions and answers satisfy a predetermined predicate function. If the players share an entangled state as part of their strategy, they can often win with higher probability~\cite{CHTW04}.

Projection games are a well-structured class of games, that have long been objects of study in the context of classical approximation algorithms. A nonlocal game is a projection game if, for any one of Alice's answers, there exists at most one correct answer Bob may give, for a given pair of questions. Projection games are known to satisfy strong parallel repetition properties~\cite{Rao08,DS14}, which also extend to the quantum setting~\cite{DSV15}.

A motivating example for projection games arises in the context of \emph{constraint systems} (CSs). A CS consists of variables and a collection of constraints over an alphabet on subsets of the variables such that, when the variables are assigned values from the alphabet, all of the constraints should be satisfied. Many natural computational problems, such as graph colouring, can be phrased as constraint systems.

There are two natural ways to present a constraint system as a nonlocal game. We refer to the first as the \emph{constraint-variable} framework, following the notation of Culf and Mastel~\cite{CM24}; it is the original framework, proposed by Cleve and Mittal~\cite{CM14}. Here, the referee samples a constraint according to a given probability distribution on the constraint system and a variable uniformly from the constraint. Alice is given the constraint, and needs to answer with a satisfying assignment; Bob is given the variable and needs to answer with an element of the alphabet. They win if Bob's assignment is consistent with Alice's. Since Bob's correct answer is the unique consistent assignment to his variable, this presentation is a projection game.

There is a second presentation, used in~\cite{Arkh12,PS25,MS24}, which we refer to as the \emph{constraint\hspace{0cm}-constraint} framework. Here, the referee samples two constraints according to a probability distribution and gives one to each player. Then, Alice and Bob have to reply with satisfying assignments to their constraint. They win if the assignments agree on all overlapping variables. This presentation is however not a projection game.

The two frameworks agree on perfect strategies, in the sense that if there is a classical, quantum, or commuting operator strategy that wins with probability $1$ in one framework, there is also a strategy that wins with probability $1$ in the other framework~\cite{CM24}. However, for general (non-perfect) strategies there are differences, which give rise to benefits and drawbacks for working in either framework. For example, unlike in the constraint-variable framework, the games in the constraint-constraint framework are symmetric between the two players. On the other hand, the constraint-variable framework provides -- via Bob's operators -- a global operator representation for every variable in the CS, unlike in the constraint-constraint framework, where there are only local representations corresponding to each constraint.

Importantly, the constraint-constraint game for a CS is a synchronous game. In a synchronous game, the question distribution is symmetric between the players, and whenever they get the same question, they win only if they respond with the same answer~\cite{PSSTW16}. This simple condition on the predicate turns out to tightly constrain the algebraic structure of perfect strategies. In particular, any synchronous quantum correlation can be expressed as a \emph{tracial correlation}, wherein Alice's and Bob's operators are modelled on the same Hilbert space with respect to a tracial state. In fact, assuming the game checks the synchronicity condition with high probability via consistency checks, almost perfect quantum correlations are close to tracial correlations~\cite{Vid22}. This has also been extended to general commuting operator correlations in infinite dimensions~\cite{MdlS23}. An alternate approach in the commuting operator model was studied in~\cite{Lin23}.

The situation is strikingly similar in the constraint-variable framework. In fact, due to~\cite{CM14,ABdSZ17}, it is known that perfect quantum strategies can also be expressed as tracial correlations. Further, almost perfect quantum strategies are close to tracial strategies, in the strong sense that the operators of an almost perfect strategy can be rounded to the operators of a tracial strategy~\cite{Pad22}. However, the error incurred by this rounding does not scale well with the question and answer size --- scaling polynomially with the number of constraints and exponentially with the number of variables per constraint.

This scaling causes issues in translating between explicit strategies for the game and settings where only working with tracial correlations is beneficial, such as the weighted algebra formalism~\cite{MS24,CM24}. One way to get around this is to synchronise the game by symmetrising the probability distribution and adding consistency checks, but this is quite inelegant and ad hoc.

A similar situation arises in the context of oracularisation. The \emph{oracularisation} of a nonlocal game transforms it into a new game where one of the two players knows the whole output statistics, in the case of a perfect strategy. This is often important in applications of nonlocal games~\cite{JNV+21}. A simple way to achieve this is by asking Alice both questions and Bob one of the two questions, chosen at random. To guarantee that the nice tracial form of the strategies is preserved, prior work has made use of a slightly more complicated version of oracularisation, which outputs a synchronous game~\cite{JNV+21,MS24}.

In this work, we extend the techniques of~\cite{MdlS23} to the setting of projection games. More precisely, we show the following:

\begin{faketheorem}[\cref{thm:main,cor:approx-finite,cor:main-degenerate}]
    Suppose that there is a commuting operator correlation that wins with probability $1-\varepsilon$ in a constraint-variable CS game. Then, there is a tracial correlation that wins with probability at least $1-O((L\varepsilon)^{1/4})$, where $L$ is the inverse of the minimal probability of Bob's questions conditioned on one of Alice's questions. Furthermore, if the original correlation is quantum (\emph{i.e.} finite-dimensional), then so is the tracial correlation.
\end{faketheorem}

This result differs from the analogous result for synchronous games only in the inverse-\hspace{0cm}polynomial dependence on the conditional question probability. For CS games, this is exactly the maximum context size. This arises from the fact that Alice's answer length may be much longer than Bob's, since the projection functions need only map from Alice's answer to Bob's. Due to this, in the proof, Alice's measurements correspond to each of Bob's possible questions need to be pieced together to give properties of the joint measurement, introducing a dependence on the probability of the least likely of Bob's questions. Whether a dependence on the context size is necessary in this setting remains an open question. However, for the applications of constraint-variable CS games we are interested in, such as constraint satisfaction problems as in~\cite{CM24}, the context size is a constant, giving winning probability $1-O(\varepsilon^{1/4})$. The same holds for oracularisations.

To prove our main result, we make use of techniques introduced in~\cite{Vid22} and expanded on in~\cite{MdlS23}. There, Alice and Bob's measurement operators are studied in the context of the von Neumann algebras they generate, where Connes' lemma~\cite{Con76} (or a generalised version in the commuting operator case~\cite{MdlS23}) is used to transform the measurements to measurements in a von Neumann algebra where the shared state becomes tracial. To be able to apply this result in the present work, we need to show that jointly-measurable near-projective measurements have a joint measurement that is near projective. This lets us round Alice and Bob's measurements to projections. This needs to be done in order to translate properties about Alice's measurement on each variable to properties of her joint measurement on each constraint.

\paragraph{Outlook} Many works on nonlocal games rely on the synchronous form of the game to simplify the structure of the perfect strategies~\cite{JNV+21,MS24,CM24}. However, we show that an analogous property holds for projection games, which are already known to satisfy strong properties such as exponential parallel repetition~\cite{DSV15}, which provides an avenue to simplify manipulations of nonlocal games.

Projection games and synchronous games provide two large classes of nonlocal games where the perfect strategies are tracial in a robust way. Are there other meaningful examples where this arises?

\paragraph{Outline} In \cref{sec:prelims}, we introduce the necessary background material from the theory of von Neumann algebras that go into this work; and in \cref{sec:nonlocal-games}, we recall the background on nonlocal games and study some basic properties of projection games. Next, in \cref{sec:perfect-q}, we extend the results of~\cite{CM14,ABdSZ17} on the structure of perfect quantum strategies to projection games. In \cref{sec:perfect-co}, we extend the results of the previous section to the setting of commuting operator correlations. In \cref{sec:approx-co}, we prove the main result for both the quantum and commuting operator setting. Finally, in \cref{sec:applications}, we discuss the consequences of the main result in the context of CS games and oracularisations.

\paragraph{Acknowledgements} The author would like to thank Archishna Bhattacharyya for helpful comments and discussions, and Connor Paddock for useful comments on an earlier version of this manuscript. The author is supported by a CGS D award from Canada's NSERC.

\section{Mathematical preliminaries}\label{sec:prelims}

\subsection{Notation}

For $m\in\N$, write $[m]=\{1,2,\ldots,m\}\subseteq\N$. We denote elements of a Cartesian product $A_1\times\cdots\times A_n$ with bold face as $\mbf{a}=(a_1,\ldots,a_n)$.

We consider only probability distributions on finite sets, in which case we can represent a probability distribution $\pi$ on $X$ by its probability density function, which we denote $\pi:X\rightarrow[0,1]$. We write the expectation value of some function $f:X\rightarrow\C$ as $\expec_{x\leftarrow\pi}f(x)=\sum_{x\in X}\pi(x)f(x)$.

For a finite-dimensional Hilbert space $H=\C^n$, we always fix a standard basis $\set*{\ket{i}}{i=1,\ldots,n}$. A \emph{positive operator valued measurement (POVM)} with (finite) set of outcomes $A$ on a Hilbert space $H$ is a collection of positive operators $\{P_{a}\}_{a\in A}\subseteq\mc{B}(H)$ such that $\sum_{a\in A}P_a=1$. $\{P_a\}_{a\in A}$ is a \emph{projection valued measure (PVM)} if additionally all the $P_a$ are projections. In that case $P_aP_{a'}=\delta_{a,a'}P_a$. We write the commutator of two elements $A,B\in\mc{B}(H)$ as $[A,B]=AB-BA$. Given finite sets $I_1,\ldots,I_n$ and a POVM $\{A_{\mbf{i}}\}_{\mbf{i}\in I_1\times\cdots\times I_n}$, we write
$$A^j_{i}=\sum_{\mbf{i}\in \prod_kI_k:\;i_j=i}A_{\mbf{i}}.$$
Note that $\{A^j_i\}_{i\in I_j}$ is a POVM. Conversely, we say POVMs $\{A^1_i\}$, ... , $\{A^n_i\}$ are \emph{jointly measurable} if there is a POVM $\{A_{\mbf{i}}\}$, called the \emph{joint measurement}, such that $A^j_i=\sum_{\mbf{i}:\;i_j=i}A_{\mbf{i}}$.

\subsection{von Neumann algebras}

A \emph{von Neumann algebra} $\mc{M}$ is a $\ast$-subalgebra of $\mc{B}(H)$ for some Hilbert space $H$ that is closed under the weak-$\ast$ topology. We always assume $\mc{H}$ is separable. Equivalently, $\mc{M}$ is a $\ast$-subalgebra that is closed under the double commutant, where the \emph{commutant} of a set $S\subseteq\mc{B}(H)$ is $S'=\set*{x\in\mc{B}(H)}{[x,s]=0\;\forall\;s\in S}$. More abstractly, a von Neumann algebra is a $C^\ast$-algebra that is the dual of some Banach space, called the \emph{predual} and denoted $\mc{M}_\ast$. As a $C^\ast$-algebra, $\mc{M}$ is a equipped with a positive cone $\mc{M}_+=\set*{x^\ast x}{x\in\mc{M}}$. 

A \emph{weight} on a von Neumann algebra is a map $\tau:\mc{M}_+\rightarrow[0,\infty]$ satisfying $\tau(x+y)=\tau(x)+\tau(y)$ and $\tau(\lambda x)=\lambda\tau(x)$ for all $x,y\in\mc{M}_+$ and $\lambda\geq 0$ (where we take $0\cdot\infty=0$). A weight is \emph{tracial} if $\tau(x^\ast x)=\tau(xx^\ast)$ for all $x\in\mc{M}$, in which case it is called a \emph{trace}. A weight is \emph{normal} if for every SOT-convergent increasing net $(x_\alpha)_\alpha$ in $\mc{M}_+$, $\lim_\alpha\tau(x_\alpha)=\tau(\lim_\alpha x_\alpha)$. A weight is \emph{faithful} if $\tau(x)=0$ iff $x=0$. A weight is \emph{finite} if $\tau(1)<\infty$; in this case $\infty\notin\im\mc{M}_+$ and $\tau$ can be extended to a positive linear functional on all $\mc{M}$. A finite weight satisfying $\tau(1)=1$ is called a \emph{state}; for a state, being normal is equivalent to having a concrete GNS representation as a state on a Hilbert space. A weight is \emph{semifinite} if for every $x\in\mc{M}_+$, there exists an increasing net $(x_\alpha)_\alpha$ in $\mc{M}_+$ such that $\lim_\alpha x_\alpha =x$ in the strong operator topology and $\tau(x_\alpha)<\infty$ for all $\alpha$.

\subsection{Noncommutative $L^p$ spaces}\label{sec:noncom-lp}

In \cref{sec:perfect-co,sec:approx-co}, we make use of Haagerup's noncommutative $L^p$ spaces~\cite{Haa79,Ter81}, as for synchronous games in~\cite{MdlS23}.

Let $\mc{M}\subseteq\mc{B}(H)$ be a von Neumann algebra. From Tomita-Takesaki theory~\cite{Tak03}, there exists a unique (up to isomorphism) von Neumann algebra $c(\mc{M})\subseteq \mc{B}(c(H))$, called the core of $\mc{M}$, equipped with a normal faithful semifinite trace $\tau$, a continuous group homomorphism $\theta:\R\rightarrow\Aut(\mc{M})$, and an embedding $\iota:\mc{M}\hookrightarrow c(\mc{M})$, satisfying $\im\iota=\set*{x\in\mc{M}}{\theta_s(x)=x\;\forall\,s\in\R}$ and $\tau\circ\theta_s=e^{-s}\tau$.

Let $\mc{N}\subseteq\mc{B}(H)$ be a von Neumann algebra with a normal faithful semifinite trace $\tau$. An unbounded operator $y$ on $H$ is \emph{affiliated} with $\mc{N}$ if $yx\subseteq xy$ for all $x\in\mc{N}'$. We call a closed densely-defined $y$ affiliated with $\mc{N}$ \emph{$\tau$-measurable} if for all $\delta>0$, there exists a projection $p\in\mc{N}$ such that $\im p\subseteq D(y)$ and $\tau(1-p)<\delta$. The set of affiliated operators is a $\ast$-algebra with respect to the strong sum and product, \emph{i.e.} standard sums and products of operators are closable and we take the strong versions of the operations to be their respective closures. It admits a Hausdorff topology with respect to which $\mc{N}$ is dense.

For $p\geq 1$, define $L^p(\mc{M})$ to be the set of $\tau$-measurable operators affiliated to $c(\mc{M})$ that are $e^{-s/p}$-eigenvectors of $\theta_s$ for all $s$. The map $x\mapsto\norm{x}_p=\tau(\chi_{(1,\infty)}(|x|))^{1/p}$ is a norm with respect to which $L^p(\mc{M})$ is complete. We have natural identifications $L^\infty(\mc{M})\cong\mc{M}$, $L^1(\mc{M})\cong\mc{M}_\ast$. For $x\in L^1(\mc{M})$, define $\tr(x)=\varphi(1)$, where $\varphi\in\mc{M}_\ast$ is the functional corresponding to $x$. $\tr$ is a linear functional on $L^1(\mc{M})$ satisfying $\tr(xy)=\tr(yx)$ for $x\in L^p(\mc{M})$ and $y\in L^q(\mc{M})$ with $\frac{1}{p}+\frac{1}{q}=1$. We can write the $p$-norm as $\norm{x}_p=\tr(|x|^p)^{1/p}$; it follows that for positive $x\in L^1(\mc{M})$, we have that $\tr(x)=\tau(\chi_{(1,\infty)}(x))$. We can define the inner product $\ang*{x,y}=\tr(x^\ast y)$ on $L^2(M)$, which makes it into a Hilbert space.

\begin{example}\label{ex:type-i}
    To illustrate the construction, consider the case $\mc{M}=\mbb{M}_n$. In this case, we can simply take $c(H)=L^2(\R)\otimes\C^n$ and $c(\mc{M})=L^\infty(\R)\otimes\mbb{M}_n$, the space of bounded measurable matrix-valued functions on the real line. In this case, $\theta_s(x)(t)=x(t-s)$ and the trace $\tau(x)=\int_{-\infty}^\infty \Tr(x(t))e^{-t}dt$, where $\Tr$ is the usual trace on $\mbb{M}_n$. The unbounded operators affiliated with $c(\mc{M})$ are the measurable functions $\R\rightarrow\mbb{M}_n$; and all of these are $\tau$-measurable. $L^p(\mc{M})$ consists of those $\tau$-measurable $x$ such that $x(t)=x(0)e^{t/p}$; hence is in this case $L^p(\mc{M})$ is isomorphic to $\mbb{M}_n$ for all $p$. The isomorphism $L^1(\mc{M})\rightarrow\mc{M}_\ast$ is given by $x\mapsto\Tr(\cdot x(0))$, and hence $\tr(x)=\Tr(x(0))$. The $p$-norms in this case reduce to the Schatten norms. The same holds for any finite von Neumann algebra.
\end{example}

\section{Nonlocal games}\label{sec:nonlocal-games}

\begin{definition}
    A \emph{nonlocal game} is a tuple $G=(I,J,\{A_i\}_{i\in I},\{B_j\}_{j\in J},\mu,V)$, where $I,J,A_i,B_j$ are finite sets, representing the question and answer sets, respectively; $\mu$ is a probability distribution on $I\times J$, representing the question distribution; and $V:\bigcup_{(i,j)\in I\times J}\{(i,j)\}\times A_i\times B_j\rightarrow\{0,1\}$ is a function, called the predicate function, representing the correct answers.

    A strategy for a nonlocal game is given by a \emph{correlation}, a function $c:\bigcup_{(i,j)\in I\times J}\{(i,j)\}\times A_i\times B_j\rightarrow[0,1]$ such that $c(\cdot,\cdot|i,j)$ is a probability distribution for all $i\in J$ and $j\in J$. 
    \begin{itemize}
        \item $c$ is a \emph{quantum correlation} if there exist finite-dimensional Hilbert spaces $H_A$ and $H_B$, POVMs $\{A^i_a\}_{a\in A_i}\subseteq\mc{B}(H_A)$ and $\{B^j_b\}_{b\in B_j}\subseteq\mc{B}(H_B)$ for all $i\in I$, $j\in J$, and a unit vector $\ket{\psi}\in H_A\otimes H_B$ such that $c(a,b|i,j)=\braket{\psi}{A^i_a\otimes B^j_b}{\psi}$. 
        
        \item $c$ is a \emph{commuting operator} correlation if there exists a Hilbert space $H$, POVMs $\{A^i_a\}_{a\in A_i}\subseteq\mc{B}(H)$ and $\{B^j_b\}_{b\in B_j}\subseteq\mc{B}(H)$, and a unit vector $\ket{\psi}\in H$ such that $[A^i_a,B^j_b]=0$ for all $i,j,a,b$ and $c(a,b|i,j)=\braket{\psi}{A^i_aB^j_b}{\psi}$. 
        
        \item $c$ is a \emph{tracial} correlation if there exists a von Neumann algebra $\mc{M}\subseteq\mc{B}(H)$, PVMs $\{A^i_a\}_{a\in A_i}\subseteq\mc{M}$ and $\{B^j_b\}_{b\in B_j}\subseteq\mc{M}$, and a normal tracial state $\tau:\mc{M}\rightarrow\C$ such that $c(a,b|i,j)=\tau(A^i_aB^j_b)$.
    \end{itemize}
    
    The \emph{winning probability} of a correlation $c$ for a nonlocal game $G$ is
    \begin{align*}
        \mfk{w}_G(c)=\expec_{(i,j)\leftarrow\mu}\sum_{a\in A_i,b\in B_j}V(a,b|i,j)c(a,b|i,j).
    \end{align*}
    We say that $c$ is a \emph{perfect strategy} if $\mfk{w}_G(c)=1$. 
\end{definition}

It is a standard fact that every tracial correlation is a commuting-operator correlation, by taking the standard form of $\mc{M}$~\cite{Haa75}. Further, if $H$ is finite-dimensional, then $c$ is also a quantum correlation.

\begin{definition}
    A nonlocal game $G=(I,J,\{A_i\}_{i\in I},\{B_j\}_{j\in J},\mu,V)$ is a \emph{projection game} if for all $i\in I$ and $j\in J$, there exists a function $p_{i,j}:A_i\rightarrow B_j$ such that $V(a,b|i,j)=\delta_{b,p_{i,j}(a)}$.
\end{definition}

Now, we introduce a simple new definition which helps us control the structure of a projection game.

\begin{definition}
    We say a projection game $G$ is \emph{non-degenerate} if the function $p_i:A_i\rightarrow\prod_{\substack{j\in J:\\\mu(i,j)>0}}B_j$ defined as $p_{i}(a)=(p_{i,j}(a))_j$ is injective for all $i\in I$.
\end{definition}

The major examples of projection games mentioned above, constraint-variable CS games and oracularisations, are always non-degenerate.

If a game is degenerate, whenever Alice answers $a$, she has the freedom to choose any element of $p_i^{-1}(p_i(a))$ without affecting the winning probability. This can cause some problems in pinpointing the behaviour of the players. However, we can always reduce to the case of non-degenerate projection games, while preserving the important properties of the game.

\begin{lemma}\label{lem:degenerate-games}
    Let $G$ be a projection game. Then, there exists a non-degenerate projection game $\tilde{G}$ such that for every correlation $c$ for $G$, there exists a correlation $\tilde{c}$ for $\tilde{G}$ such that $\mfk{w}_{G}(c)=\mfk{w}_{\tilde{G}}(\tilde{c})$, and vice versa. Further, the mappings from $c$ to $\tilde{c}$ and from $\tilde{c}$ to $c$ preserve the classes of quantum, commuting operator, and tracial correlations.
\end{lemma}

Though it is not relevant to this work, the mappings also preserve every other natural (closed under classical post-processing) class of correlations by construction, notably the classical correlations, the non-signalling correlations, and the levels of the NPA hierarchy.

\begin{proof}
    Let $\tilde{A}_i\subseteq A_i$ be a set of representatives for the partition of $A_i$ induced by $p_i$ --- that is, for every $a\in A_i$, there exists a unique $\tilde{a}\in A_i$ such that $p_{i}(\tilde{a})=p_{i}(a)$. Then, we can define $\tilde{V}$ as the restriction of $V$ to $\bigcup_{i,j}\{(i,j)\}\times \tilde{A}_i\times B_j$, and the nonlocal game $\tilde{G}=(I,J,\{\tilde{A}_i\}_{i\in I},\{B_j\}_{j\in J},\mu,\tilde{V})$. By definition, this is a projection game via the projection functions $\tilde{p}_{i,j}=p_{i,j}|_{\tilde{A}_i}$. Now, let $c$ be a correlation for $G$. Define $\tilde{c}$ as the correlation $\tilde{c}(a,b|i,j)=\sum_{a'\in p_i^{-1}(p_i(a))}c(a',b|i,j)$. Since $\tilde{c}$ is defined from $c$ via classical post-processing on Alice's system, the mapping preserves quantum, commuting operator, and tracial correlations. It also preserves the winning probability as
    \begin{align*}
        \mfk{w}_{\tilde{G}}(\tilde{c})&=\expec_{(i,j)\leftarrow\mu}\sum_{a\in \tilde{A}_i, b\in B_j}\tilde{V}(a,b|i,j)\tilde{c}(a,b|i,j)\\
        &=\expec_{(i,j)\leftarrow\mu}\sum_{a\in \tilde{A}_i, b\in B_j}\tilde{c}(a,p_{i,j}(a)|i,j)\\
        &=\expec_{(i,j)\leftarrow\mu}\sum_{a\in \tilde{A}_i, b\in B_j}\sum_{a'\in p_i^{-1}(p_i(a))}c(a',p_{i,j}(a)|i,j)\\
        &=\expec_{(i,j)\leftarrow\mu}\sum_{a'\in A_i, b\in B_j}c(a',p_{i,j}(a')|i,j)\\
        &=\mfk{w}_G(c).
    \end{align*}
    For the other direction, let $\tilde{c}$ be a correlation for $\tilde{G}$. Then, define the correlation $c$ for $G$ as $c(a,b|i,j)=\tilde{c}(a,b|i,j)$ when $a\in\tilde{A}_i$ and $c(a,b|i,j)=0$ otherwise. As before, it is easy to see that this mapping preserves quantum, commuting-operator, and tracial correlations. Finally it preserves the winning probability as
    \begin{align*}
        \mfk{w}_{G}(c)&=\expec_{(i,j)\leftarrow\mu}\sum_{a\in A_i,b\in B_j}V(a,b|i,j)c(a,b|i,j)\\
        &=\expec_{(i,j)\leftarrow\mu}\sum_{a\in \tilde{A}_i,b\in B_j}\tilde{V}(a,b|i,j)\tilde{c}(a,b|i,j)\\
        &=\mfk{w}_{\tilde{G}}(\tilde{c}).
    \end{align*}
\end{proof}

For the majority of the paper, we will assume that we are dealing with non-degenerate projection games, but using the above lemma, we can guarantee that the main result also holds for degenerate games.

\section{Perfect quantum case}\label{sec:perfect-q}

In this section, we extend the result of \cite{CM14,ABdSZ17} showing that perfect correlations for constraint-variable CS games are tracial to the setting of general projection games. We express the proofs using the von Neumann algebra language that we make use of throughout this work.

\begin{lemma}[\emph{e.g.} \cite{ABdSZ17}]\label{lem:simplifying-q-corrs}
    Let $c$ be a quantum correlation. Then, there exist a finite-\hspace{0cm}dimensional Hilbert space $H=\C^d$, POVMs $\{P^i_a\}_{a\in A_i}\subseteq\mc{B}(H)$ and $\{Q^j_b\}_{b\in B_j}\subseteq\mc{B}(H)$, and a unit vector $\ket{\phi}=\sum_{n=1}^d\lambda_n\ket{n}\ket{n}\in H\otimes H$ with $\lambda_n>0$ such that $c(a,b|i,j)=\braket{\phi}{P^i_a\otimes Q^j_b}{\phi}$.
\end{lemma}

\begin{proof}
    Since $c$ is a quantum correlation, there exist some finite-dimensional Hilbert spaces $H_A,H_B$, POVMs $\{A^i_a\}_{a\in A_i}\subseteq\mc{B}(H_A)$, $\{B^j_b\}_{b\in B_j}\subseteq\mc{B}(H_B)$, and unit vector $\ket{\psi}\in H_A\otimes H_B$ such that $c(a,b|i,j)=\braket{\psi}{A^i_a\otimes B^j_b}{\psi}$. Let $\ket{\psi}=\sum_{n=1}^d\lambda_n\ket{\psi_n}\ket{\phi_n}$ be the Schmidt decomposition of $\ket{\psi}$, for $d$ the Schmidt rank of $\ket{\psi}$. Then, for $H=\C^d$, define the isometries $U:H\rightarrow H_A$ and $V:H\rightarrow H_B$ as $U=\sum_n\ketbra{\psi_n}{n}$ and $V=\sum_n\ketbra{\phi_n}{n}$. Letting $P^i_a=U^\ast A^i_a U$ and $Q^j_b=V^\ast B^j_b V$, we get that
    \begin{align*}
        &c(a,b|i,j)=\braket{\psi}{A^i_a\otimes B^j_b}{\psi}=\braket{\phi}{(U\otimes V)^\ast(A^i_a\otimes B^j_b)(U\otimes V)}{\phi}=\braket{\phi}{P^i_a\otimes Q^j_b}{\phi}.\qedhere
    \end{align*}
\end{proof}

In the main theorem of this section, we make use of the following result.

\begin{proposition}[\cite{HRS08}]\label{prop:joint-measurable}
    Let $\{A^1_i\}$, ... , $\{A^n_i\}$ be jointly measurable PVMs. Then the joint measurement $A_{i_1...i_n}=A^1_{i_1}\cdots A^n_{i_n}$ and $\{A_{i_1,\ldots,i_n}\}_{i_1,\ldots,i_n}$ is a PVM.
\end{proposition}

Note that the above proposition holds independently of dimension.

\begin{theorem}\label{thm:perfect-q}
    Let $G=(I,J,\{A_i\}_{i\in I},\{B_j\}_{j\in J},\mu,V)$ be a non-degenerate projection game, and let $c(a,b|i,j)=\braket{\phi}{P^i_a\otimes Q^j_b}{\phi}$ be a quantum correlation for $G$ of the form guaranteed by \cref{lem:simplifying-q-corrs}. Suppose $c$ is a perfect strategy for $G$. Define $D=\sum_{n=1}^d\lambda_n\ketbra{n}$. Then, we have that
    \begin{itemize}
        \item $P^i_a$ and $Q^j_b$ are projections
        \item $D$ commutes with $P^i_a$ and $Q^j_b$
        \item $P^i_a (Q^j_b)^T=0$ whenever $p_{i,j}(a)\neq b$.
    \end{itemize}
\end{theorem}

This theorem was shown in the case of constraint-variable BCS games in~\cite{CM14} and for constraint-variable CS games over a general alphabet  in~\cite{ABdSZ17}.

\begin{proof}
    For $c$ to be perfect, we need that $c(a,b|i,j)=0$ whenever $V(a,b|i,j)=0$. Hence, $\braket{\phi}{P^i_a\otimes Q^j_b}{\phi}=0$ if $p_{i,j}(a)\neq b$. In this case, we have $\sqrt{P^i_a\otimes Q^j_b}\ket{\phi}=0$, which implies $(P^i_a\otimes Q^j_b)\ket{\phi}=0$. Rewriting this as a matrix gives $P^i_a D(Q^j_b)^T=0$. Write $P^{i,j}_b=\sum_{a\in p_{i,j}^{-1}(b)}P^i_a$. By summing over all $b\neq p_{i,j}(a)$, we find that $P^i_a D(1-(Q^j_{b})^T)=0$ if $b=p_{i,j}(a)$. Then, summing over $a$ such that $p_{i,j}(a)=b$, we find that $P^{i,j}_bD(Q^j_b)^T=P^{i,j}_bD$. On the other hand, summing over those $a$ such that $p_{i,j}(a)\neq b$, we can get $(1-P^{i,j}_b)D(Q^j_b)^T$. Together $P^{i,j}_bD(Q^j_b)^T=P^{i,j}_bD=D(Q^j_b)^T$.
    
    Then, $(P^{i,j}_b)^2D=P^{i,j}_bD(Q^j_b)^T=P^{i,j}_bD$, so by invertibility of $D$, $(P^{i,j}_b)^2=P^{i,j}_b$. As such, $P^{i,j}_b$ is a projection. For each $x$, we define a measurement on $\prod_{j\in J:\,\mu(i,j)>0}B_j$ as $P^i_{\mbf{b}}=\sum_{a\in p_{i}^{-1}(\mbf{b})}P^i_{a}$. By construction, this is the joint measurement of the $\{P^{i,j}_b\}$ over $j$. Using \cref{prop:joint-measurable}, $P^i_{\mbf{b}}$ is a projection and $P^i_{\mbf{b}}=\prod_{j}P^{i,j}_{b_j}$. Since $G$ is non-degenerate, $P^i_{\mbf{b}}$ is the sum of at most $1$ term, so we get that $P^i_{a}=\prod_{j}P^{i,j}_{p_{i,j}(a)}$ and is a projection. In the same way, $D((Q^j_b)^2)^T=P^{i,j}_bD(Q^j_b)^T=D(Q^j_b)^T$, so $(Q^y_b)^T$ and hence $Q^y_b$ is a projection as well.

    Next, $P^{i,j}_bD^2=D(Q^j_b)^TD=D(D(Q^j_b)^T)^\ast=D^2P^{i,j}_b$, so $P^{i}_a$ commutes with $D^2$. By the continuous functional calculus, $P^i_a$ commutes with $(D^2)^{1/2}=D$. In the same way, $Q^j_b$ commutes with $D$.

    To finish, note that if $p_{i,j}(a)\neq b$, $0=P^i_a D (Q^j_b)^T=P^i_a(Q^j_b)^TD$. By invertibility of $D$, $P^i_a(Q^j_b)^T=0$.
\end{proof}

Note that, if $G$ were degenerate, we could not rebuild $P^x_a$ from the $P^{x,y}_b$. Hence $P^x_a$ might not be a projection and might not commute with $D$. Nevertheless the other results would hold.

\begin{corollary}
    Let $G$ be a non-degenerate projection game, and let $c$ be a quantum correlation. Suppose that $c$ is a perfect strategy for $G$. Then, $c$ is a tracial correlation.
\end{corollary}

This corollary follows for the case of constraint-variable CS games from \cite{ABdSZ17,Pad22,CM14}.

\begin{proof}
    By \cref{thm:perfect-q}, we know that $c(a,b|i,j)=\Tr(P^i_a D(Q^j_b)^T D)$, where $D>0$ is invertible and $\Tr(D^2)=1$, $P^i_a$ and $Q^j_b$ are projectors, $D$ commutes with $P^i_a$ and $Q^j_b$. First note that $c(a,b|i,j)=\Tr(P^i_a(Q^j_b)^T D^2)$. Now, let $\mc{M}\subseteq\mc{B}(H)$ be the von Neumann algebra generated by the $P^i_a$ and $(Q^j_b)^T$. By construction, we have that $D\in\mc{M}'$. Let $\tau$ be the state on $\mc{M}$ defined as $\tau(x)=\Tr(xD^2)$. As $D^2\in\mc{M}'$, $\tau$ is tracial. Let $A^i_a=P^i_a\in\mc{M}$ and $B^j_b=(Q^j_b)^T\in\mc{M}$. We have that $c(a,b|i,j)=\tau(A^i_a B^j_b)$, so it is a tracial correlation.
\end{proof}

\section{Perfect commuting operator case}\label{sec:perfect-co}

Using ideas from \cref{sec:noncom-lp}, it is possible to extend the results of the previous section from quantum correlations to commuting-operator correlations.

\begin{theorem}
    Let $G=(I,J,\{A_i\}_{i\in I},\{B_j\}_{j\in J},\mu,V)$ be a non-degenerate projection game, and let $c(a,b|i,j)=\braket{\psi}{P^i_a Q^j_b}{\psi}$ be a commuting operator correlation that is perfect for $G$. Then, $c$ is a tracial correlation.
\end{theorem}

\begin{proof}
    Let $\mc{M}\subseteq\mc{B}(H)$ be the von Neumann algebra generated by the $Q^j_b$. Then $P^i_a\in\mc{M}'$. The map $\mc{M}\rightarrow\C$, $x\mapsto\braket{\psi}{x}{\psi}$ is an element of $\mc{M}_\ast$, so there exists a corresponding element $h_\psi\in L^1(\mc{M})$. Note that $h_\psi$ is positive and $\tr(h_\psi)=\psi(1)=1$. By a standard argument, there exist POVMs $\{\bar{P}^i_a\}_{a\in A_i}\subseteq\mc{M}$ such that $\braket{\psi}{P^i_a Q^j_b}{\psi}=\tr(\bar{P}^i_ah_\psi^{1/2}Q^j_bh_\psi^{1/2})$ (see the proof of Theorem 2.1 in \cite{MdlS23}). As $c$ is perfect, we know that $c(a,b|i,j)=0$ whenever $p_{i,j}(a)=b$. By tracial property of $\tr$, we have in this case that
    \begin{align*}
        0=\tr((\bar{P}^i_a)^{1/2}h_\psi^{1/2}Q^j_bh_\psi^{1/2}(\bar{P}^i_a)^{1/2})=\norm{(\bar{P}^i_a)^{1/2}h_\psi^{1/2}(Q^j_b)^{1/2}}_2^2,
    \end{align*}
    and hence $(\bar{P}^i_a)^{1/2}h_\psi^{1/2}(Q^j_b)^{1/2}=0$ (in terms of the strong product). Therefore, $\bar{P}^i_ah_\psi^{1/2}Q^j_b=0$. Let $\bar{P}^{i,j}_b=\sum_{a\in p_{i,j}^{-1}(b)}\bar{P}^i_a$. Fixing any $b\neq b'\in B_j$, we have $\bar{P}^{i,j}_{b}h_\psi^{1/2}Q^j_{b'}=\sum_{a\in p_{i,j}^{-1}(b')}\bar{P}^i_ah_\psi^{1/2}Q^j_b=0$ as $p_{i,j}(a)=b'\neq b$. Now, $\bar{P}^{i,j}_bh_\psi^{1/2}=\sum_{b'\in B_j}\bar{P}^{i,j}_bh_\psi^{1/2}Q^j_{b'}=\bar{P}^{i,j}_bh_\psi^{1/2}Q^j_b$ and $h_\psi^{1/2}Q^x_b=\sum_{b'\in B_j}\bar{P}^{i,j}_{b'}h_\psi^{1/2}Q^j_b=\bar{P}^{i,j}_bh_\psi^{1/2}Q^j_b$. Now, in the same way as in the finite-dimensional case, $(\bar{P}^{i,j}_b)^2h_\psi^{1/2}=\bar{P}^{i,j}_bh_\psi^{1/2}Q^j_b=\bar{P}^{i,j}_bh_\psi^{1/2}$; $h_\psi^{1/2}(Q^j_b)^2=\bar{P}^{i,j}_bh_\psi^{1/2}Q^j_b=h_\psi^{1/2}Q^j_b$; and $Q^j_bh_\psi=(h_\psi^{1/2}Q^j_b)^\ast h_\psi^{1/2}=h_\psi^{1/2}\bar{P}^{i,j}_bh_\psi^{1/2}=h_\psi Q^j_b$. As the $Q^j_b$ generate $\mc{M}$, $h_\psi$ is affiliated with $\mc{M}'$. Let $p$ be the projection onto the support of $h_\psi$. Then, $p\mc{M}\subseteq \mc{B}(p c(H))$ is a von Neumann algebra, and $p^{i,j}_b=\bar{P}^{i,j}_bp$ and $q^j_b=Q^j_bp$ are projections in $p\mc{M}$. Since the $p^{i,j}_b$ are jointly measurable and $G$ is non-degenerate, we have that $p^i_a=\prod_{j\in J}p^{i,j}_{p_{i,j}(a)}$ is a PVM as in \cref{thm:perfect-q}; note that $p^i_a=\bar{P}^i_a p$. Now, let $\rho:p\mc{M}\rightarrow\C$ be defined as $\rho(x)=\tr(xh_\psi)$. $\rho$ is a tracial state: as $\rho(1_{p\mc{M}})=\rho(p)=\tr(h_\psi)=1$, for $x\geq 0$, $\rho(x)=\tr(h_\psi^{1/2}xh_\psi^{1/2})\geq0$, and $\rho(xy)=\tr(xyh_{\psi})=\tr(xh_{\psi}y)=\tr(yxh_{\psi})=\rho(yx)$. Putting it together,
    \begin{align*}
        c(a,b|i,j)=\tr(p^i_a h_\psi^{1/2} q^j_b h_\psi^{1/2})=\tr(p^i_a q^j_b h_\psi)=\rho(p^i_a q^j_b),
    \end{align*}
    which is a tracial correlation.
\end{proof}

\begin{corollary}
    In the notation of the previous proof, we have for a perfect commuting operator strategy that $p^i_a q^j_b=0$ whenever $p_{i,j}(a)\neq b$.
\end{corollary}

\section{Approximate case}\label{sec:approx-co}

In this section, we prove the following theorem, extending to the most general case of almost perfect commuting operator strategies.

\begin{theorem}\label{thm:main}
    Let $G=(I,J,\{A_i\}_{i\in I},\{B_j\}_{j\in J},\mu,V)$ be a non-degenerate projection game, and let $c(a,b|i,j)=\braket{\phi}{P^i_aQ^j_b}{\phi}$ be a quantum correlation for $G$. Suppose that $\mfk{w}_{G}(c)\geq 1-\varepsilon$. Then, there exists a tracial correlation $c'$ such that $\mfk{w}_{G}(c')\geq 1-148(L\varepsilon)^{1/4}$, where $L=\max_{i,j:\pi(i,j)\neq 0}\frac{\pi(i)}{\pi(i,j)}$ for $\pi(i)=\sum_{j\in J}\pi(i,j)$.
\end{theorem}

Our main technical tool is the following lemma of \cite{MdlS23}, which they use to show the analogous result for synchronous games.

\begin{lemma}[\cite{MdlS23}]\label{lem:mdls}
    Let $\mc{M}\subseteq\mc{B}(H)$ be a von Neumann algebra, $h\in L^1(\mc{M})$ be positive such that $\tr(h)=1$, and $q=\chi_{(1,\infty)}(h)$. Let $X$ and $A$ be finite sets, and let $\pi$ be a symmetric probability distribution on $X\times X$ with marginal $\nu$ on $X$. Suppose that $\{p^x_a\}_{a\in A}$ are PVMs such that $\expec_{x\leftarrow\nu}\sum_a\norm{[p^x_a,h^{1/2}]}_2^2\leq\delta$. Then, there exist PVMs $\{r^x_a\}_{a\in A}\subseteq q c(\mc{M})q$ such that
    \begin{align*}
        \expec_{(x,y)\leftarrow\pi}\sum_{a,b\in A}\abs*{\tr(p^x_ah^{1/2}p^y_bh^{1/2})-\tau(r^x_ar^y_b)}\leq 38\delta^{1/4}.
    \end{align*}
\end{lemma}

Note that $\tau|_{qc(\mc{M})q}$ is a tracial state as $\tau(q)=\tr(h)=1$.

We also need the following theorem to allow us to round POVMs to PVMs

\begin{theorem}[\cite{dlS22}]\label{thm:dls-orth}
    Let $\rho$ be a normal state on a von Neumann algebra $\mc{M}$. Suppose that $\{a_i\}\subseteq\mc{M}$ is a POVM such that $\sum_i\rho(a_i^2)\geq 1-\varepsilon$. Then, there exists a PVM $\{p_i\}\subseteq\mc{M}$ such that $\sum_i\rho(|a_i-p_i|^2)\leq9\varepsilon$.
\end{theorem}

The following lemma shows that (with respect to KMS norm) if jointly measurable POVMs are approximately PVMs, then their joint measurement is also approximately a PVM. Note that this result is the only place where the dependence on the number of variables per constraint in \cref{thm:main} arises.

\begin{lemma}\label{lem:approx-jointly-meas}
    Let $\mc{M}\subseteq\mc{B}(H)$ be a von Neumann algebra and let $h\in L^1(\mc{M})$ be a positive element such that $\tr(h)=1$. Let $\pi$ be a strictly positive probability distribution of $[n]$. Suppose that $\{A_{\mbf{i}}\}_{\mbf{i}\in[k]^n}\subseteq\mc{M}$ is a POVM such that $\sum_{j=1}^n\pi(j)\sum_{i=1}^k\norm{h^{1/4}A^j_ih^{1/4}}_2^2\geq1-\delta$. Then,
    $$ \sum_{\mbf{i}\in[k]^n}\norm{h^{1/4}A_{\mbf{i}}h^{1/4}}_2^2\geq1-M\delta,$$
    where $M=\max_i\frac{1}{\pi(i)}$.
\end{lemma}

Before passing to the proof, note that this bound is tight in the limit of small $\delta$. Take $k=2$, $\pi$ uniform, $\mc{M}=\mc{B}(\C^2)^{\otimes n}$ with tracial state $\frac{1}{2^n}\Tr$, and $A_{\mbf{i}}=\bigotimes_{j}\parens*{(1-\varepsilon)\ketbra{i_j}+\frac{\varepsilon}{2}I}$. Then, for any $j$, $$\sum_{i=1,2}\frac{1}{2^n}\Tr((A^j_i)^2)=\frac{1}{2}\sum_{i=1,2}\Tr\parens*{((1-\varepsilon)\ketbra{i}+\tfrac{\varepsilon}{2}I)^2}=\frac{1}{2}\sum_{i=1,2}\Tr\parens*{(1-\varepsilon)\ketbra{i}+\tfrac{\varepsilon^2}{4}I}=1-\varepsilon+\tfrac{\varepsilon^2}{2},$$
giving $\delta=\varepsilon-\frac{\varepsilon^2}{2}$. In the same way, we have that
$$\sum_{\mbf{i}}\frac{1}{2^n}\Tr(A_{\mbf{i}})=\parens*{\frac{1}{2}\sum_{i=1,2}\Tr\parens*{((1-\varepsilon)\ketbra{i}+\tfrac{\varepsilon}{2}I)^2}}^n=(1-\delta)^n\leq 1-n\delta+\binom{n}{2}\delta^2.$$
For small $\delta\ll\frac{1}{n}$, this is approximately $1-n\delta$.

\begin{proof}
    Write $\delta_j=1-\sum_{i=1}^k\norm{h^{1/4}A^j_ih^{1/4}}_2^2$. Then, $$\sum_{i}\norm{h^{1/4}\sum_{\mbf{i}\in[k]^n:\,i_j=i}A_{\mbf{i}}h^{1/4}}_2^2=\sum_{i}\norm{h^{1/4}A_{i}^jh^{1/4}}_2^2\geq 1-\delta_j.$$ Hence,
    \begin{align*}
        \delta_j&\geq\norm[\Big]{\sum_{\mbf{i}}h^{1/4}A_{\mbf{i}}h^{1/4}}_2^2-\sum_{i}\norm[\Big]{\sum_{\mbf{i}:\,i_j=i}h^{1/4}A_{\mbf{i}}h^{1/4}}_2^2\\
        &=\sum_{\mbf{i},\mbf{i}'}\tr(A_{\mbf{i}}h^{1/2}A_{\mbf{i}'}h^{1/2})-\sum_{\substack{\mbf{i},\mbf{i}':\,i_j=i_j'}}\tr(A_{\mbf{i}}h^{1/2}A_{\mbf{i}'}h^{1/2})\\
        &=\sum_{\mbf{i},\mbf{i}':\,i_j\neq i_j'}\tr(A_{\mbf{i}}h^{1/2}A_{\mbf{i}'}h^{1/2}).
    \end{align*}
    Putting these together via a union bound,
    \begin{align*}
        \sum_{\mbf{i}}\norm[\big]{h^{1/4}A_{\mbf{i}}h^{1/4}}_2^2&=1-\sum_{\mbf{i},\mbf{i}':\,\mbf{i}\neq\mbf{i}'}\tr(A_{\mbf{i}}h^{1/2}A_{\mbf{i}'}h^{1/2})\\
        &\geq1-\sum_j\sum_{\mbf{i},\mbf{i}':\,i_j\neq i_j'}\tr(A_{\mbf{i}}h^{1/2}A_{\mbf{i}'}h^{1/2})\\
        &\geq1-\sum_j\delta_j\geq1-M\sum_j\pi(j)\delta_j\geq1-M\delta.\qedhere
    \end{align*}
\end{proof}

\begin{proof}[Proof of \cref{thm:main}]
    We may also assume that $\varepsilon\leq\frac{1}{L}$, as otherwise the result holds trivially.
    
    As in the previous section, we take $\mc{M}\subseteq\mc{B}(H)$ to be the von Neumann algebra generated by the $Q^j_b$. Then, there exist a positive $h_\psi\in L^1(\mc{M})$ and POVMs $\{\bar{P}^i_a\}_{a\in  A_i}\subseteq\mc{M}$ such that $\braket{\psi}{P^i_a Q^j_b}{\psi}=\tr(\bar{P}^i_a h_\psi^{1/2}Q^j_b h_\psi^{1/2})$. Therefore, writing $\bar{P}^{i,j}_b=\sum_{a\in p_{i,j}^{-1}(b)}\bar{P}^i_a$
    \begin{align*}
        \mfk{w}_{G}(c)&=\sum_{i,j}\mu(i,j)\sum_{a\in A_i}\braket{\psi}{P^i_a Q^j_{p_{i,j}(a)}}{\psi}\\
        &=\sum_{i,j}\mu(i,j)\sum_{a\in A_i}\tr(\bar{P}^i_a h_\psi^{1/2}Q^j_{p_{i,j}(a)} h_\psi^{1/2})\\
        &=\sum_{i,j,b}\mu(i,j)\ang*{h_{\psi}^{1/4}\bar{P}^{i,j}_b h_\psi^{1/4}, h_\psi^{1/4}Q^j_b h_\psi^{1/4}}.
    \end{align*}
    From here, we often make use of three Hilbert spaces: $V_1=L^2(\mc{M})^{\bigcup_{j\in J}I\times\{j\}\times B_j}$ equipped with the inner product $$\ang*{(u_{i,j,b})_{i,j,b},(v_{i,j,b})_{i,j,b}}_{V_1}=\sum_{i,j,b}\mu(i,j)\ang*{u_{i,j,b},v_{i,j,b}},$$ and $V_2=L^2(\mc{M})^{\bigsqcup_{i\in I}\{i\}\times J\times A_i}$ equipped with the inner product $$\ang*{(u_{i,j,a})_{i,j,a},(v_{i,j,a})_{i,j,a}}_{V_2}=\sum_{i,j,a}\mu(i,j)\ang*{u_{i,j,a},v_{i,j,a}}.$$ Consider $(h_{\psi}^{1/4}\bar{P}^{i,j}_b h_\psi^{1/4})_{i,j,b}$ and $(h_{\psi}^{1/4}Q^{j}_b h_\psi^{1/4})_{i,j,b}$ as vectors in $V_1$. First, these vectors have norm at most $1$. In fact,
    \begin{align*}
        \norm{(h_{\psi}^{1/4}\bar{P}^{i,j}_b h_\psi^{1/4})_{i,j,b}}^2_{V_1}&=\sum_{i,j,b}\pi(i,j)\tr(\bar{P}^{i,j}_bh_{\psi}^{1/2}\bar{P}^{i,j}_b h_\psi^{1/2})\\
        &\leq\sum_{i,j,b}\mu(i,j)\tr(h_{\psi}^{1/2}\bar{P}^{i,j}_b h_\psi^{1/2})\\
        &\leq\sum_{i,j}\mu(i,j)\tr(h_{\psi})=1,
    \end{align*}
    and identically $\norm{(h_{\psi}^{1/4}Q^{j}_b h_\psi^{1/4})_{i,j,b}}_{V_1}\leq 1$. By hypothesis, we know that $$\ang*{(h_{\psi}^{1/4}\bar{P}^{i,j}_b h_\psi^{1/4})_{i,j,b},(h_{\psi}^{1/4}Q^{j}_b h_\psi^{1/4})_{i,j,b}}_{V_1}\geq 1-\varepsilon,$$ so by Cauchy-Schwarz, $$1-\varepsilon\leq\norm{(h_{\psi}^{1/4}\bar{P}^{i,j}_b h_\psi^{1/4})_{i,j,b}}_{V_1}\norm{(h_{\psi}^{1/4}Q^{j}_b h_\psi^{1/4})_{i,j,b}}_{V_1}\leq\min\!\!\set*{\norm{(h_{\psi}^{1/4}\bar{P}^{i,j}_b h_\psi^{1/4})_{i,j,b}}_{V_1},\norm{(h_{\psi}^{1/4}Q^{j}_b h_\psi^{1/4})_{i,j,b}}_{V_1}}.$$ This gives that the commutator
    \begin{align*}
        \sum_{i,j,b}\mu(i,j)\norm[\big]{[Q^j_b,h_\psi^{1/2}]}_2^2&=\sum_{i,j,b}\mu(i,j)\tr\parens*{(Q^j_bh_\psi^{1/2}-h_\psi^{1/2}Q^j_b)(h_\psi^{1/2}Q^j_b-Q^j_bh_\psi^{1/2})}\\
        &=2\sum_{i,j,b}\mu(i,j)\tr\parens*{(Q^j_b)^2h_\psi-h_\psi^{1/2}Q^j_bh_\psi^{1/2}Q^j_b}\\
        &\leq 2-2\norm{(h_\psi^{1/4}Q^j_bh_\psi^{1/4})_{i,j,b}}_{V_1}^2\leq 2-2(1-\varepsilon)^2\leq 4\varepsilon.
    \end{align*}
    Next, from the fact that $\sum_{i,j,b}\mu(i,j)\norm[\big]{h_\psi^{1/4}\bar{P}^{i,j}_bh_\psi^{1/4}}_2^2\geq(1-\varepsilon)^2\geq1-2\varepsilon$, we can use \cref{lem:approx-jointly-meas} to get that $\sum_{i\in I,j\in J,a\in A_i}\mu(i,j)\norm[\big]{h_\psi^{1/4}\bar{P}^i_a h_\psi^{1/4}}_2^2\geq1-2L\varepsilon$. Hence, it follows identically that $$\sum_{i,j,a}\mu(i,j)\norm[\big]{[\bar{P}^i_a,h_\psi^{1/2}]}_2^2\leq 4L\varepsilon.$$ From here, it follows that
    \begin{align*}
        1-2\varepsilon&\leq\sum_{i,j,b}\mu(i,j)\tr(Q^j_bh_\psi^{1/2}Q^j_bh_\psi^{1/2})\\
        &=\sum_{i,j,b}\mu(i,j)\ang*{h_\psi^{1/2}Q^j_b,Q^j_bh_\psi^{1/2}}\\
        &\leq\sum_{i,j,b}\mu(i,j)\norm{h_\psi^{1/2}Q^j_b}_2\norm{Q^j_bh_\psi^{1/2}}_2\\
        &=\sum_{i,j,b}\mu(i,j)\tr((Q^j_b)^2h_\psi).
    \end{align*}
    In the same way, $$\sum_{i,j,a}\mu(i,j)\tr((\bar{P}^i_a)^2h_{\psi})\geq1-2L\varepsilon.$$ Since for all $x\in\mc{M}$, $\tr(xh_\psi)=\braket{\psi}{x}{\psi}$, the map $x\mapsto\tr(\cdot h_\psi)$ is a normal state on $\mc{M}$. Hence, by \cref{thm:dls-orth}, there exist PVMs $\{p^i_a\}_{a\in A_i},\{q^j_b\}_{b\in B_j}\subseteq\mc{M}$ such that
    \begin{align*}
        &\norm{((Q^j_b-q^j_b)h_\psi^{1/2})_{i,j,b}}_{V_1}^2=\sum_{i,j,b}\mu(i,j)\tr\parens*{\parens*{Q^j_b-q^j_b}^2h_\psi}\leq 18\varepsilon,\\
        &\norm{((\bar{P}^i_a-p^i_a)h_\psi^{1/2})_{i,j,a}}_{V_2}^2=\sum_{i,j,a}\mu(i,j)\tr(\parens*{\bar{P}^i_a-p^i_a}^2h_\psi)\leq 18L\varepsilon.
    \end{align*}
    Now, using again the norm on $V_1$, we get the approximate commutation
    \begin{align*}
        \sum_{i,j,b}\mu(i,j)\norm[\big]{[q^j_b,h_\psi^{1/2}]}_2^2&=\norm{([q^j_b,h_\psi^{1/2}])_{i,j,b}}_{V_1}^2\\
        &=\norm[\big]{([Q^j_b,h_\psi^{1/2}])_{i,j,b}-((Q^j_b-q^j_b)h_\psi^{1/2})_{i,j,b}+(h_\psi^{1/2}(Q^j_b-q^j_b))_{i,j,b}}_{V_1}^2\\
        &\leq\parens*{\norm[\big]{([Q^j_b,h_\psi^{1/2}])_{i,j,b}}_{V_1}+2\norm[\big]{((Q^j_b-q^j_b)h_\psi^{1/2})_{i,j,b}}_{V_1}}^2\\
        &\leq(\sqrt{4\varepsilon}+2\sqrt{18\varepsilon})^2\leq110\varepsilon
    \end{align*}
    and in the same way using $V_2$, $\sum_{i,j,a}\mu(i,j)\norm[\big]{[p^i_a,h_\psi^{1/2}]}_2^2\leq110L\varepsilon$. To finish the proof, we want to make use of \cref{lem:mdls}. Define $X=I\sqcup J$ and $A=\sqcup_{i,j} A_i\sqcup B_j$. Take $\pi$ to be the probability distribution on $X\times X$ defined as
    \begin{align*}
        \pi(x,x')=\begin{cases}\frac{1}{2}\mu(x,x')&x\in I\land x'\in J\\\frac{1}{2}\mu(x',x)&x'\in I\land x\in J\\0&\text{else}\end{cases}.
    \end{align*}
    Note that $\pi$ is symmetric and the marginal of $\pi$ on one copy of $X$ is
    \begin{align*}
        \nu(x)=\begin{cases}\frac{1}{2}\sum_{j\in J}\mu(x,j)&x\in I\\\frac{1}{2}\sum_{i\in I}\mu(i,x)&x\in J\end{cases}.
    \end{align*}
    Let $\tilde{p}^x_a=p^x_a$ if $x\in I$ and $\tilde{p}^x_a=q^x_a$ if $x\in J$. By the above, we have that
    \begin{align*}
        \expec_{x\leftarrow\nu}\sum_a\norm{[\tilde{p}^x_a,h_\psi^{1/2}]}_2^2\leq110L\varepsilon.
    \end{align*}
    Then, using \cref{lem:mdls}, we get that there exist PVMs $\{r^x_a\}_{a\in A}\subseteq qc(\mc{M})q$ such that $$\expec_{(x,y)\leftarrow\pi}\sum_{a,b\in A}\abs*{\tr(\tilde{p}^x_ah_\psi^{1/2}\tilde{p}^y_bh_\psi^{1/2})-\tau(r^x_ar^y_b)}\leq38(110L\varepsilon)^{1/4}\leq124(L\varepsilon)^{1/4},$$ which implies
    \begin{align*}
        \sum_{i,j,a}\mu(i,j)\abs*{\tr(p^i_a h_\psi^{1/2}q^j_{p_{i,j}(a)}h_\psi^{1/2})-\tau(r^i_a r^j_{p_{i,j}(a)})}\leq124(L\varepsilon)^{1/4}.
    \end{align*}
    Now, note that
    \begin{align*}
        \sum_{i,j,a}\mu(i,j)\tr\parens[\big]{\parens[\big]{\bar{P}^i_a-(\bar{P}^i_a)^{1/2}}^2h_\psi}&=\sum_{i,j,a}\mu(i,j)\tr\parens[\big]{\parens[\big]{(\bar{P}^i_a)^2+\bar{P}^i_a-2(\bar{P}^i_a)^{3/2}}h_\psi}\\
        &\leq\sum_{i,j,a}\mu(i,j)\tr\parens[\big]{\parens[\big]{\bar{P}^i_a-(\bar{P}^i_a)^2}h_\psi}\\
        &\leq1-(1-2L\varepsilon)=2L\varepsilon,
    \end{align*}
    and in the same way, $\sum_{i,j,b}\mu(i,j)\tr\parens*{(Q^j_b-(Q^j_b)^{1/2})^2h_\psi}\leq 2\varepsilon$. Then, we can bound
    \begin{align*}
        &\abs[\Big]{\sum_{i,j,a}\mu(i,j)\tr(\bar{P}^i_a h_\psi^{1/2}Q^j_{p_{i,j}(a)}h_\psi^{1/2})-\sum_{i,j,a}\mu(i,j)\tr(p^i_a h_\psi^{1/2}q^j_{p_{i,j}(a)}h_\psi^{1/2})}\\
        &=\abs*{\norm[\big]{((\bar{P}^i_a)^{1/2} h_\psi^{1/2}(Q^j_{p_{i,j}(a)})^{1/2})_{i,j,a}}_{V_2}^2-\norm[\big]{(p^i_a h_\psi^{1/2}q^j_{p_{i,j}(a)})_{i,j,a}}_{V_2}^2}\\
        &\leq2\abs*{\norm[\big]{((\bar{P}^i_a)^{1/2} h_\psi^{1/2}(Q^j_{p_{i,j}(a)})^{1/2})_{i,j,a}}_{V_2}-\norm[\big]{(p^i_a h_\psi^{1/2}q^j_{p_{i,j}(a)})_{i,j,a}}_{V_2}}\\
        &\leq2\norm[\big]{((\bar{P}^i_a)^{1/2} h_\psi^{1/2}(Q^j_{p_{i,j}(a)})^{1/2})_{i,j,a}-(p^i_a h_\psi^{1/2}q^j_{p_{i,j}(a)})_{i,j,a}}_{V_2}\\
        &\leq2\norm[\big]{((p^i_a-(\bar{P}^i_a)^{1/2})h_{\psi}^{1/2}(Q^j_{p_{i,j}(a)})^{1/2})_{i,j,a}}_{V_2}+2\norm[\big]{(p^i_a h_\psi^{1/2}(q^j_{p_{i,j}(a)}-(Q^j_{p_{i,j}(a)})^{1/2}))_{i,j,a}}_{V_2}\\
        &=2\sqrt{\sum_{i,j,a}\mu(i,j)\tr((p^i_a-(\bar{P}^i_a)^{1/2})^2h_{\psi}^{1/2}Q^j_{p_{i,j}(a)}h_\psi^{1/2})}+2\sqrt{\sum_{i,j,b}\mu(i,j)\tr(p^{i,j}_b h_\psi^{1/2}(q^j_{b}-(Q^j_{b})^{1/2})^2h_\psi^{1/2})}
   	\end{align*}
	
	\begin{align*}
        &\leq2\sqrt{\sum_{i,j,a}\mu(i,j)\tr((p^i_a-(\bar{P}^i_a)^{1/2})^2h_{\psi})}+2\sqrt{\sum_{i,j,b}\mu(i,j)\tr((q^j_{b}-(Q^j_{b})^{1/2})^2h_\psi)}\hspace{3.5cm}\\
        &=2\norm[\big]{((p^i_a-(\bar{P}^i_a)^{1/2})h_{\psi}^{1/2})_{i,j,a}}_{V_2}+2\norm[\big]{((q^j_{b}-(Q^j_{b})^{1/2})h_\psi^{1/2})_{i,j,b}}_{V_1}\\
        &\leq2\norm[\big]{((p^i_a-\bar{P}^i_a)h_{\psi}^{1/2})_{i,j,a}}_{V_2}+\norm[\big]{((\bar{P}^i_a-(\bar{P}^i_a)^{1/2})h_{\psi}^{1/2})_{i,j,a}}_{V_2}+2\norm[\big]{((q^j_{b}-Q^j_{b})h_\psi^{1/2})_{i,j,b}}_{V_1}\\
        &\qquad\qquad+2\norm[\big]{((Q^j_{b}-(Q^j_{b})^{1/2})h_\psi^{1/2})_{i,j,b}}_{V_1}\\
        &\leq2\sqrt{18L\varepsilon}+2\sqrt{2L\varepsilon}+2\sqrt{18\varepsilon}+2\sqrt{2\varepsilon}\\
        &\leq23(L\varepsilon)^{1/2}.
    \end{align*}
    Via the triangle inequality, we get that
    \begin{align*}
        \abs[\Big]{\sum_{i,j,a}\mu(i,j)\parens*{\braket{\psi}{P^i_a Q^j_{p_{i,j}(a)}}{\psi}-\tau(r^i_a r^j_{p_{i,j}(a)})}}\leq23(L\varepsilon)^{1/2}+124(L\varepsilon)^{1/4}\leq147(L\varepsilon)^{1/4}.
    \end{align*}
    To finish, take $\{\hat{p}^i_a\}_{a\in A_i}$ to be the completion of $\{r^i_a\}_{a\in A_i}$ to a PVM by adding $1-\sum_{a\in A_i}r^i_a$ to an arbitrary element; and construct $\{\hat{q}^j_b\}_{b\in B_j}$ from $\{r^j_b\}_{b\in B_j}$ in the same way. Then $c'(a,b|i,j)=\tau(\hat{p}^i_a\hat{q}^j_b)$ is a tracial correlation and we find that
    \begin{align*}
        \mfk{w}_{G}(c')&=\sum_{i,j}\mu(i,j)\sum_{a\in A_i}\tau(\hat{p}^i_a\hat{q}^j_{p_{i,j}(a)})\\
        &\geq \sum_{i,j}\mu(i,j)\sum_{a\in A_i}\tau(r^i_a r^j_{p_{i,j}(a)})\\
        &\geq\mfk{w}_{G}(c)-\abs[\big]{\sum_{i,j,a}\mu(i,j)\parens*{\braket{\psi}{P^i_a Q^j_{p_{i,j}(a)}}{\psi}-\tau(r^i_a r^j_{p_{i,j}(a)})}}\\
        &\geq1-148(L\varepsilon)^{1/4}.\qedhere
    \end{align*}
\end{proof}

\begin{corollary}\label{cor:approx-finite}
    Let $G$ be as in the theorem. Suppose $c(a,b|i,j)=\braket{\psi}{P^i_a\otimes Q^j_b}{\psi}$ is a quantum correlation such that $\mfk{w}_{G}(c)\geq 1-\varepsilon$. Then, there exists a finite-dimensional tracial correlation $c'$ such that $\mfk{w}_{G}(c')\geq 1-148(L\varepsilon)^{1/4}$.
\end{corollary}

\begin{proof}
    Take $\mc{M}$ to be the von Neumann algebra generated by the $1\otimes Q^j_b$; this is a finite-dimensional von Neumann algebra contained in $\mc{B}(H_A\otimes H_B)$. Then, following \cref{ex:type-i}, $c(\mc{M})=L^\infty(\R)\otimes\mc{M}$. From the theorem, there exists a correlation $c'(a,b|i,j)=\tau(p^i_a q^j_b)$ such that $\mfk{w}_{G}(c')\geq 1-148(L\varepsilon)^{1/4}$, where $\tau$ is the semifinite weight $\tau(x)=\int_{-\infty}^\infty\Tr(x(t))e^{-t}dt$ on $c(\mc{M})$, and $p^i_a,q^j_b\in qc(\mc{M})q$ for some projection $q\in c(\mc{M})$ satisfying $\tau(q)=1$. By construction, $q(t),p^i_a(t),q^j_b(t)$ are projections for every $t\in\R$, so
    \begin{align*}
        \tau(p^i_a q^j_b)=\int_{-\infty}^\infty\Tr(p^i_a(t) q^j_b(t))e^{-t}dt=\int_{\set*{t}{\Tr(q(t))\neq 0}}\frac{1}{\Tr(q(t))}\Tr(p^i_a(t) q^j_b(t))\Tr(q(t))e^{-t}dt
    \end{align*}
    is a convex combination of finite-dimensional tracial correlations. This implies that at least one of the correlations has winning probability $\geq 1-148(L\varepsilon)^{1/4}$, completing the proof.
\end{proof}

Now, we show that the same result can be recovered for degenerate games.

\begin{corollary}\label{cor:main-degenerate}
    Suppose $G$ is an arbitrary projection game. Suppose there exists a commuting operator correlation $c$ for $G$ such that $\mfk{w}_G(c)\geq 1-\varepsilon$. Then, there exists a tracial correlation $c'$ for $G$ such that $\mfk{w}_{G}(c')\geq 1-148(L\varepsilon)^{1/4}$. If $c$ is quantum, then $c'$ can be chosen to be finite-dimensional.
\end{corollary}

\begin{proof}
    If $G$ is non-degenerate, this is exactly the result of \cref{thm:main,cor:approx-finite}. Now, suppose $G$ is degenerate. Let $\tilde{G}$ be the non-degenerate game constructed in \cref{lem:degenerate-games}. By that lemma, there exists a commuting operator correlation $\tilde{c}$ such that $\mfk{w}_{\tilde{G}}(\tilde{c})=\mfk{w}_{G}(c)\geq1-\varepsilon$. Then, by \cref{thm:main}, there exists a tracial correlation $\tilde{c}'$ for $\tilde{G}$ such that $\mfk{w}_{\tilde{G}}(\tilde{c}')\geq 1-148(L\varepsilon)^{1/4}$. Using again \cref{lem:degenerate-games}, there exists a tracial correlation $c'$ for $G$ such that $\mfk{w}_G(c')=\mfk{w}_{\tilde{G}}(\tilde{c}')\geq1-148(L\varepsilon)^{1/4}$. The result for quantum correlations holds using \cref{cor:approx-finite}.
\end{proof}

\section{Applications}\label{sec:applications}

In this section, we study the implications of our main result for CS games and oracularisations,

\subsection{Constraint system games}

\begin{definition}
    A \emph{constraint} over an alphabet $\Sigma$ is a pair $(V,C)$ where $V$ is a finite set, called the set of variables or the context, and $C\subseteq\Sigma^V$. We sometimes call $C$ the constraint, when $V$ is clear; we see $C\subseteq\Sigma^n$ as a constraint on $[n]$. A \emph{constraint system (CS)} over an alphabet $\Sigma$ is a pair $S=(X,\{(V_i,C_i)\}_{i=1}^m)$ where $X$ is a finite set, called the set of variables, and the $(V_i,C_i)$ are constraints over $\Sigma$ such that $V_i\subseteq X$.
\end{definition}

We always work with finite alphabets. A classical satisfying assignment to the constraint system $S$ is a function $f:X\rightarrow\Sigma$ such that $f|_{V_i}\subseteq C_i$ for each $i\in[m]$. We present quantum assignments via nonlocal games. Here, we use the nonsynchronous version of the game, first introduced in~\cite{CM14}.

\begin{definition}
    Given a constraint system $S=(X,\{(V_i,C_i)\}_{i=1}^m)$ over $\Sigma$ and a probability distribution $\pi$ on $[m]$, the \emph{constraint-variable CS game} $G(S,\pi)$ is the nonlocal game with input sets $I=[m]$ and $J=X$, output sets $A_i=C_i$ and $B_j=\Sigma$, and question distribution and predicate functions
    \begin{align*}
        \mu(i,x)&=\begin{cases}\frac{\pi(i)}{|V_i|}&x\in V_i\\0&\text{else}\end{cases}\\
        V(\phi,\sigma|i,x)&=\begin{cases}1&\phi(x)=\sigma\\0&\text{else}\end{cases}.
    \end{align*}
\end{definition}
We may without loss of generality assume that $\pi$ is strictly positive; else, some of the constraints are never asked, and hence we can remove them from the CS.

$G(S,\pi)$ is always a projection game via the projection function $p_{i,x}(\phi)=\phi(x)$.

\begin{corollary}
    Let $S=(X,\{(V_i,C_i)\}_{i=1}^m)$ be a constraint system over $\Sigma$ and let $\pi$ be probability distribution over $[m]$. If $c$ is a commuting operator correlation for $G(S,\pi)$ such that $\mfk{w}_{G(S,\pi)}(c)\geq 1-\varepsilon$, then there exists a tracial correlation $c'$ for $G(S,\pi)$ such that $\mfk{w}_{G(S,\pi)}(c')\geq 1-148(L\varepsilon)^{1/4}$, where $L=\max_i|V_i|$. If $c$ is quantum, $c'$ can be chosen to be finite-dimensional.
\end{corollary}

\begin{proof}
    This follows directly from \cref{thm:main,cor:approx-finite} by noting that
    \begin{align*}
        L&=\max_{i,x:\,\mu(i,x)\neq 0}\frac{\mu(i)}{\mu(i,x)}=\max_{i}\frac{\pi(i)}{\frac{\pi(i)}{|V_i|}}=\max_i|V_i|.\qedhere
    \end{align*}
\end{proof}

This also implies that results in the weighted algebra setting, which considers only tracial states, can now be translated directly to the game setting, without having to work with a synchronised version of the game with consistency checks.

\begin{corollary}
    Let $S$, $\pi$, and $L$ be as in previous corollary. Let $\mc{A}_{c-v}(S,\pi)$ be the constraint-variable weighted algebra defined in~\cite{CM24}.
    \begin{enumerate}[1.]
        \item If there exists a commuting operator correlation $c$ such that $\mfk{w}_{G(S,\pi)}(c)\geq1-\varepsilon$, then there exists a tracial state $\tau$ on $\mc{A}_{c-v}(S,\pi)$ such that $\defect(\tau)\leq 148(L\varepsilon)^{1/4}$.
        \item If there exists a quantum correlation $c$ such that $\mfk{w}_{G(S,\pi)}(c)\geq1-\varepsilon$, then there exists a finite-dimensional tracial state $\tau$ on $\mc{A}_{c-v}(S,\pi)$ such that $\defect(\tau)\leq 148(L\varepsilon)^{1/4}$.
    \end{enumerate}
\end{corollary}

\subsection{Oracularisations}

\begin{definition}
    Let $G=(I,J,\{A_i\}_{i\in I},\{B_j\}_{j\in J},\mu,V)$ be a nonlocal game. The (projection) \emph{oracularisation} of $G$ is the nonlocal game $G^{\orac}=(I\times J,I\sqcup J,\{A_i\times B_j\}_{(i,j)\in I\times J},\{A_i\}_{i\in I}\cup\{B_j\}_{j\in J},\mu^{\orac},V^{\orac})$, where
    $$\mu^{\orac}((i,j),k)=\begin{cases}\frac{1}{2}\mu(i,j)&k=i\lor k=j\\0&\text{else,}\end{cases}$$
    and $$V^{\orac}((a,b),c|(i,j),k)=\begin{cases}1&V(a,b|i,j)=1\land k=i\land c=a\\1&V(a,b|i,j)=1\land k=j\land c=b\\0&\text{else.}\end{cases}$$
\end{definition}

Oracularisations are clearly projection games via the projection functions $p_{(i,j),i}((a,b))=a$ and $p_{(i,j),j}((a,b))=b$. Prior work~\cite{JNV+21,MS24} has made use of a more involved synchronous version of oracularisation, but we show that this is not necessary.

\begin{corollary}
    Let $G$ be a nonlocal game and let $c$ be a commuting operator correlation for $G^{\orac}$. If $\mfk{w}_{G^{\orac}}(c)\geq 1-\varepsilon$, then there exists a tracial correlation $c'$ for $G^{\orac}$ such that $\mfk{w}_{G^{\orac}}(c')\geq 1-177\varepsilon^{1/4}$. If $c$ is quantum, then $c'$ can be chosen to be finite-dimensional.
\end{corollary}

\begin{proof}
    This follows directly from \cref{thm:main,cor:approx-finite} by noting that $$L=\max_{i,j,k:\,\mu^{\orac}((i,j),k)\neq 0}\frac{\mu^{\orac}((i,j))}{\mu^{\orac}((i,j),k)}=\max_{i,j:\,\mu(i,j)\neq 0}\frac{\mu(i,j)}{\frac{1}{2}\mu(i,j)}=2$$ and then that $148(2)^{1/4}\leq 177$.
\end{proof}

\bibliographystyle{bibtex/bst/alphaarxiv.bst}
\bibliography{bibtex/bib/full.bib,bibtex/bib/quantum.bib,bibtex/quantum_new.bib}

\newcommand{\etalchar}[1]{$^{#1}$}
\makeatletter\@ifundefined{url}{\newcommand{\url}[1]{\texttt{#1}}}{}\@ifundefined{href}{\newcommand{\href}[2]{\texttt{#2}}}{}\@ifundefined{mathbb}{\newcommand{\mathbb}[1]{#1}}{}\makeatother\providecommand{\DobbendeBruyn}{van
  Dobben de Bruyn} \providecommand{\Do}{Dob}
\begin{thebibliography}{CHTW04}

\bibitem[ABdSZ17]{ABdSZ17}
S.~Abramsky, R.~S. Barbosa, N.~de~Silva, and O.~Zapata.
\newblock {The Quantum Monad on Relational Structures}.
\newblock In K.~G. Larsen, H.~L. Bodlaender, and J.-F. Raskin, editors, {\em
  42nd International Symposium on Mathematical Foundations of Computer Science
  (MFCS 2017)}, volume~83 of {\em Leibniz International Proceedings in
  Informatics (LIPIcs)}, pages 35:1--35:19, Dagstuhl, Germany, 2017. Schloss
  Dagstuhl -- Leibniz-Zentrum f{\"u}r Informatik.
\newblock \\
  \texttt{DOI:\,\href{http://dx.doi.org/10.4230/LIPIcs.MFCS.2017.35}{10.4230/LIPIcs.MFCS.2017.35}}.

\bibitem[Ark12]{Arkh12}
A.~Arkhipov.
\newblock Extending and characterizing quantum magic games.
\newblock {\em arXiv preprint arXiv:1209.3819}, 2012.

\bibitem[CHTW04]{CHTW04}
R.~Cleve, P.~H{\o}yer, B.~Toner, and J.~Watrous.
\newblock Consequences and limits of nonlocal strategies.
\newblock In {\em 19th Annual Conference on Computational Complexity---CCC
  2004}, pages 236--249, 2004.
\newblock \\
  \texttt{DOI:\,\href{http://dx.doi.org/10.1109/CCC.2004.1313847}{10.1109/CCC.2004.1313847}}.

\bibitem[CM14]{CM14}
R.~Cleve and R.~Mittal.
\newblock Characterization of binary constraint system games.
\newblock In {\em Automata, Languages, and Programming: 41st International
  Colloquium, ICALP 2014, Copenhagen, Denmark, July 8-11, 2014, Proceedings,
  Part I 41}, pages 320--331. Springer, 2014.

\bibitem[CM24]{CM24}
E.~Culf and K.~Mastel.
\newblock {RE}-completeness of entangled constraint satisfaction problems,
  2024.
\newblock \\ Online: \url{https://arxiv.org/abs/2410.21223}.

\bibitem[Con76]{Con76}
A.~Connes.
\newblock Classification of injective factors cases {$\textrm{II}_\lambda$},
  {$\textrm{II}_\infty$}, {$\textrm{III}_\lambda$}, {$\lambda \not= 1$}.
\newblock {\em Annals of Mathematics}, 104(1): 73--115, 1976.
\newblock \\
  \texttt{DOI:\,\href{http://dx.doi.org/10.2307/1971057}{10.2307/1971057}}.

\bibitem[dlS22]{dlS22}
M.~de~la Salle.
\newblock Orthogonalization of positive operator valued measures.
\newblock {\em Comptes Rendus. Math{\'e}matique}, 360(G5): 549--560, 2022.

\bibitem[DS14]{DS14}
I.~Dinur and D.~Steurer.
\newblock Analytical approach to parallel repetition.
\newblock In {\em Proceedings of the forty-sixth annual ACM symposium on Theory
  of computing}, pages 624--633, 2014.

\bibitem[DSV15]{DSV15}
I.~Dinur, D.~Steurer, and T.~Vidick.
\newblock A parallel repetition theorem for entangled projection games.
\newblock {\em computational complexity}, 24(2): 201--254, 2015.

\bibitem[Haa75]{Haa75}
U.~Haagerup.
\newblock The standard form of von {N}eumann algebras.
\newblock {\em Mathematica Scandinavica}, 37(2): 271--283, 1975.

\bibitem[Haa79]{Haa79}
U.~Haagerup.
\newblock Lp-spaces associated with an arbitrary von {N}eumann algebra.
\newblock In {\em Algebres d’op{\'e}rateurs et leurs applications en physique
  math{\'e}matique (Proc. Colloq., Marseille, 1977)}, volume 274, pages
  175--184, 1979.

\bibitem[HRS08]{HRS08}
T.~Heinosaari, D.~Reitzner, and P.~Stano.
\newblock Notes on joint measurability of quantum observables.
\newblock {\em Foundations of Physics}, 38: 1133--1147, 2008.

\bibitem[JNV{\etalchar{+}}21]{JNV+21}
Z.~Ji, A.~Natarajan, T.~Vidick, J.~Wright, and H.~Yuen.
\newblock {MIP}* =~{RE}.
\newblock {\em Communications of the ACM}, 64(11): 131--138, 2021.
\newblock \\
  \texttt{DOI:\,\href{http://dx.doi.org/10.1145/3485628}{10.1145/3485628}}.

\bibitem[Lin23]{Lin23}
J.~Lin.
\newblock Tracial embeddable strategies: Lifting mip* tricks to mipco.
\newblock {\em arXiv preprint arXiv:2304.01940}, 2023.

\bibitem[MdlS23]{MdlS23}
A.~Marrakchi and M.~de~la Salle.
\newblock Almost synchronous correlations and {T}omita-{T}akesaki theory.
\newblock {\em arXiv preprint arXiv:2307.08129}, 2023.

\bibitem[MS24]{MS24}
K.~Mastel and W.~Slofstra.
\newblock Two prover perfect zero knowledge for {MIP*}.
\newblock In {\em Proceedings of the 56th Annual ACM Symposium on Theory of
  Computing}, pages 991--1002, 2024.

\bibitem[Pad22]{Pad22}
C.~Paddock.
\newblock Rounding near-optimal quantum strategies for nonlocal games to
  strategies using maximally entangled states.
\newblock {\em arXiv preprint arXiv:2203.02525}, 2022.

\bibitem[PS25]{PS25}
C.~Paddock and W.~Slofstra.
\newblock Satisfiability problems and algebras of boolean constraint system
  games.
\newblock {\em Illinois Journal of Mathematics}, 69(1): 81--107, 2025.

\bibitem[PSS{\etalchar{+}}16]{PSSTW16}
V.~I. Paulsen, S.~Severini, D.~Stahlke, I.~G. Todorov, and A.~Winter.
\newblock Estimating quantum chromatic numbers.
\newblock {\em Journal of Functional Analysis}, 270(6): 2188--2222, 2016.

\bibitem[Rao08]{Rao08}
A.~Rao.
\newblock Parallel repetition in projection games and a concentration bound.
\newblock In {\em Proceedings of the Fortieth Annual ACM Symposium on Theory of
  Computing}, pages 1--10, 2008.

\bibitem[Tak03]{Tak03}
M.~Takesaki.
\newblock {\em Theory of operator algebras II}, volume 125.
\newblock Springer, 2003.

\bibitem[Ter81]{Ter81}
M.~Terp.
\newblock Lp spaces associated with von {N}eumann algebras.
\newblock {\em Notes, Math. Institute, Copenhagen Univ}, 3(4): 5, 1981.

\bibitem[Vid22]{Vid22}
T.~Vidick.
\newblock Almost synchronous quantum correlations.
\newblock {\em Journal of mathematical physics}, 63(2), 2022.

\end{thebibliography}

\end{document}